\documentclass[aps,pra,twocolumn,floatfix,
nofootinbib,showpacs,longbibliography]{revtex4-1}
\usepackage[british]{babel}  
\usepackage[sc,osf]{mathpazo}
\usepackage[scaled=0.86]{berasans}  
\usepackage[colorlinks=true, citecolor=blue, urlcolor=blue]{hyperref}  
\usepackage{graphicx} 
\usepackage[babel]{microtype}  
\usepackage{amsmath,amssymb,amsthm,bm,amsfonts,mathrsfs,bbm} 

\usepackage{xspace}  
\usepackage{pgf,tikz}
\usepackage{xcolor}
\usepackage{multirow}
\usepackage{array}
\usepackage{bigstrut}
\usepackage{braket}
\usepackage{color}
\usepackage{natbib}
\usepackage{amsmath}
\usepackage{multirow}
\usepackage{mathtools}
\usepackage{float}
\usepackage[caption = false]{subfig}
\usepackage{xcolor,colortbl}
\usepackage{color}
\newcommand{\Tr}{\operatorname{Tr}}

\newcommand{\be}{\begin{equation}}
\newcommand{\ee}{\end{equation}}
\newcommand{\ba}{\begin{eqnarray}}
\newcommand{\ea}{\end{eqnarray}}
\newcommand{\ketbra}[2]{|#1\rangle \langle #2|}

\newtheorem{theorem}{Theorem}

\newtheorem{proposition}{Proposition}

\newtheorem{lemma}{Lemma}

\begin{document}

\title{Some {\it no-go} results in Quantum Thermodynamics}

\author{Tamal Guha}
\email{g.tamal91@gmail.com}    
\affiliation{Physics and Applied Mathematics Unit, Indian Statistical Institute, 203 B.T. Road, Kolkata 700108, India.}

\author{Mir Alimuddin}
\email{aliphy80@gmail.com}    
\affiliation{Physics and Applied Mathematics Unit, Indian Statistical Institute, 203 B.T. Road, Kolkata 700108, India.}

\author{Preeti Parashar}
\email{parashar@isical.ac.in}    
\affiliation{Physics and Applied Mathematics Unit, 
	Indian Statistical Institute, 
	203 B.T. Road, Kolkata 700108, India.}

\begin{abstract}
Thermodynamics is one of the fascinating branches of traditional physics to certify the occurrence of many natural processes. On the other hand, quantum theory is the most acceptable description of the microscopic world. In the present work, we have studied how the structure of quantum theory prohibits cloning or masking of several thermodynamic quantities, viz., work and energy stored in a quantum state. Our results have important consequences in quantum partial cloning, quantum masking and on the action of quantum channels. 
\end{abstract}



\maketitle
\section{Introduction}
The laws of thermodynamics play a crucial role in characterizing the (im)possibility of executing a physical process. More precisely, any hypothetical process can be rejected if it is in disagreement with these laws. The first law is simply the energy conservation principle, whereas the second law specifies the direction of any spontaneous process \cite{zemansky}. On the other hand, the rich algebraic structure of quantum mechanics prohibits the execution of several information processing tasks \cite{nocloning, nobroadcasting, nomasking, nohiding, nonunitary, nosuperposition, tsirelsonbound, IC, nohypersignaling}. The implication of these {\it no-go} theorems as a consequence of thermodynamic principles is a field of recent interest \cite{vonNeumann, Peres, APappa'PRA, Brandao'NAT}. For example, in \cite{APappa'PRA} it is shown that Cirel'son bound of Bell-nonlocality ($B_{\mathcal{Q}}\leq2\sqrt{2}$) for quantum correlations is deeply connected to the Landaure's principle of erasing information in thermodynamics. However, it is also interesting to ask what are the restrictions imposed by quantum mechanics itself in the context of thermodynamic quantities, viz., work, internal energy, free energy etc.? In other words, how the inherent algebraic structure of quantum theory characterizes thermodynamic (im)possibilities? In this paper we have studied whether "work" and "energy" stored in a quantum state can be cloned, split or masked.\par 
Due to the presence of correlations in finite particle systems the obvious question arises about the validity of traditional thermodynamic laws in the quantum regime. It is found in the literature that these laws need to be modified in order to be relevant at the microscopic level \cite{popescu'Nat,popescu'ref,brandao'PNAS,horo'lim,bera'Nat} and to formulate an appropriate resource theory for the same \cite{brandao'PRL}. The amount of extractable work is another issue of great interest in this domain. There are mainly three different kinds of quantities associated with extractable work of a quantum state, depending upon three different accessibility scenarios. {\it First} one quantifies the average amount of extractable work using single copy of the state under unitary operations \cite{allahvardyan}. This can be asymptotically extended using large number of copies, which transforms the initial state to the corresponding same entropic minimum energetic state \cite{alicki}.
 {\it Second} one considers a single copy of a state along with a bath at lower temperature under the action of global unitary jointly on the system, bath and work qudit and transforms the system state to the bath state\cite{popescu'bath}. In the {\it third} case a single copy of the system state evolves unitarily with an assistance of a constant temperature bath to extract Renyi-$0$ free energy amount of work \cite{Horodecki'Nature}, and as an asymptotic extension the extracted work is exactly equal to the difference in free energy of the initial and final state \cite{brandao'PNAS}.  However, in this paper we are focusing on the first kind of work value stored in a quantum state. Access to single particle system prevents us to extract free energy amount of work from the closed quantum system for $d\ge3$ (which actually manifests the difference between first and last kind of extractable work) and leads to an idea of passive states \cite{lenard,skrzypzyckPRE} that contains least energy for the given spectrum.\par 
The von Neumann entropy of a quantum system makes a connection between the information content and the amount of extractable work of the state. So it is interesting to ask whether the impossibility to clone the information content in a quantum state prohibits the copying of energy or work content of the state. Here we have shown that although it is possible to clone the amount of energy content for any arbitrary quantum state, cloning of the work content is strictly prohibited. We have also shown that cloning of an energy storage is equivalent to its broadcasting. As in classical thermodynamics, it is also possible for a quantum state to split its energy content in two distinct states in any arbitrary ratio. For this we shall use energy-preserving unitary operations following the resource theoretic framework of quantum thermodynamics. Another important question in this regime is regarding the {\it gap} between locally and globally extractable amount of work. This quantity namely the {\it ergotropy gap} plays a significant role in certifying quantum entanglement present in a bipartite system \cite{huber'PRX, our'PRA}. In the extreme case of non-zero ergotropy gap the locally extractable work is zero whereas the global work is non-zero. This means the work is actually {\it masked} in the correlations of the bipartite system. We have shown that for a restricted class of states, it is possible to mask the work content of a quantum system in bipartite correlations without any thermodynamic cost. However, in general there is no universal work masking unitary, even if we allow some thermodynamical cost for its implementation.  
\section{Preliminaries}
\subsection{Evolution of the quantum state}
A general quantum state in the operator space over $d$-dim Hilbert space can be written as 
\begin{equation}
\rho=\frac{1}{d}(\mathbb{I}+\sum_{k=1}^{d^{2}-1} r_{k} \sigma_{k}), 
\end{equation}\label{1}
\\where $r_{k}$'s are the Bloch vector components and $\sigma_{k}$'s are the generalized Pauli matrices obeying Tr$[\sigma_{j}\sigma_{k}]=\delta_{jk}$ and $[\sigma_{j},\sigma_{k}]=i \epsilon_{jkl} \sigma_{l}$, with $\epsilon_{jkl}$ as the structure constants of the $SU(d)$ algebra. In a similar fashion one can define the generalized Hamiltonian as $H=\frac{1}{d}(n_{0}\mathbb{I}+\sum_{k=1}^{d^{2}-1}n_{k}\sigma_{k})$, where $\{n_{0},n_{k}\}\in \mathbb{R}$. Without loss of generality it is possible to fix the Hamiltonian axis along any of the chosen generalized Pauli matrix. Mathematically,
\begin{equation}\label{2}
 [H,\sigma_{k}]=0, 
 \end{equation}
\\where the $d$-eigenvectors of $\sigma_{k}$ denote the energy eigenstates corresponding to the energy eigenvalues of the Hamiltonian. The time evolution equation for the quantum state $\frac{\partial \rho}{\partial t}=-\frac{i}{\hbar}[H,\rho]$ can be visualized as a rotation along the $\sigma_{k}$ axis \cite{Hamiltonian'GPT}. Consider the unit Bloch Ball representation corresponding to the qubit state space ({\it see Fig\ref{fig1}}). For this case, if we associate a three dimensional vector $\vec{\rho}$ with any operator $\rho$ on the corresponding Hilbert space as $\vec{\rho}=\frac{1}{2}(r_{x},r_{y},r_{z})^{T}$ (without the co-efficient of $\mathbb{I}$), then clearly the operator $-i[H,\rho]$ has the vector representation as $(\epsilon_{jk1}n_{j}r_{k},\epsilon_{jk2}n_{j}r_{k},\epsilon_{jk3}n_{j}r_{k})^{T}\equiv \vec{H}\times\vec{\rho}$. Hence the evolution equation takes the vector form $\frac{\partial\vec{\rho}}{\partial t}\sim\vec{H}\times\vec{\rho}$. This is analogous to the rotational mechanics, in the sense that the position vector (here $\vec{\rho}$) rotates around angular velocity vector ($\vec{H}$) to generate its time evolution, i.e., the velocity ($\frac{\partial\vec{\rho}}{\partial t}$).
 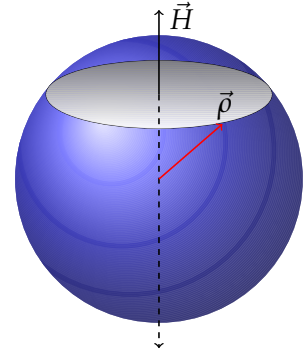
\begin{figure}[htb]
	\scalebox{0.75}{
		\begin{tikzpicture}
		\shade[ball color = blue!70, opacity = 0.5] (0,0) circle (2.55cm);
		\draw[rotate=180,black](0,-1.5)circle(2cm and 0.6cm);
		\shade[rotate=180,color = red!30](0,-1.5)circle(2cm and 0.6cm);
		\draw[->,dashed,thick](0,1.5)--(0,0)--(0,-3);
		\draw[->,thick](0,1.5)--(0,3);
		\draw [->,red!100,thick](0,0)--(1.125,0.975);
		\node at (1.17,1.3) {\Large$\vec{\rho}$};
		\node at (0.4,2.9) {\Large$\vec{H}$};
			\end{tikzpicture}}
	\caption{(Color on-line) Time evolution of a pure state vector $\vec{\rho}$ along the circular path about the vector $\vec{H}$ is depicted on the surface of the Bloch Sphere. All the states residing on the disk will have constant energy.}
	\label{fig1}
\end{figure}
    \subsection{Energy of a quantum state}
The amount of energy stored in a quantum state is given by $E(\rho)=Tr(\rho H)$ which can be simplified as $\frac{1}{d}(n_{0}+\sum_{k=1}^{d^{2}-1} r_{k}n_{k})$ due to the {\it traceless} nature of all the generalized Pauli matrices. However, using the simplified form in Eq.\eqref{2} the energy will be
\begin{equation}\label{3}
E(\rho)= Tr(\rho H)=\frac{1}{d}(n_{0}+n_{k}r_{k}).
\end{equation}
 Hence the energy of a quantum state $\rho$ in $\mathbb{R}^{d^{2}-1}$ Bloch ball is the projection of $\vec{\rho}$ along the $\vec{H}$ direction. It is clear that Eq.\eqref{3} puts a constraint on the $d^{2}-1$ free parameters of a general qudit state. Hence, all the $d$-dimensional quantum states having same energy is lie on a $d^{2}-2$ dimensional hyperplane inscribed in a $\mathbb{R}^{d^{2}-1}$ Bloch ball. It is easy to show that the {\it constant energetic} hyperplane inscribed inside the generalized Bloch ball is {\it convex} in nature. More precisely, if the $d-$dimensional quantum states $\rho$ and $\sigma$ both have equal amount of energy, i.e., $Tr(\rho H)= Tr(\sigma H)= E$, then for any $\tau=p\rho+(1-p)\sigma$, where $0\leq p\leq1$, the energy of $\tau$ can be written as, $E(\tau)=Tr(\tau H)=Tr((p\rho+(1-p)\sigma)H)=pE(\rho)+(1-p)E(\sigma)=E$ by using linearity of {\it trace}. This exhibits the convex structure of equi-energetic quantum states. However, it is worth noticing that another important thermodynamic quantity, the free energy of a quantum state, involves both energy and entropy, where the latter prevents the set of states with equal free energy to be convex in nature.
 In Fig.\ref{fig1} states having same energy lie on the entire hyperplane whereas states of equal free energy reside only on the boundary of the hyperplane. 
 \subsection{Extraction of work}
 Work extraction from a quantum state is a subject of primary focus in thermodynamics. When the concerned system is isolated from the universe, the time evolution of the operator expectation value, i.e., 
 \begin{equation}\label{4}
 \frac{d\langle \hat{A}\rangle}{dt}=\frac{i}{\hbar}\langle[\hat{H},\hat{A}]\rangle+\langle\frac{\partial \hat{A}}{\partial t}\rangle
 \end{equation}
  demands that the energy of the state evolving under time independent Hamiltonian should be preserved. However, the action of a properly chosen time varying potential $V(t)$ can lower its energy. Obviously the evolution is unitary as the applied potential acts only on the system state. So the time dependent Hamiltonian can be written as, $\vec{H}^{'}=\vec{H}+\vec{V}(t)$ under which the state will evolve to reach the corresponding passive state. If the spectral form of the Hamiltonian is given by $H=\sum_{k=1}^{n}\epsilon_{k}\ket{k}\bra{k}$, then for the state $\rho=\sum_{k=1}^{n}p_{k}\ket{\psi_{k}}\bra{\psi_{k}}$ the passive form will be $\rho_{p}=\sum_{k=1}^{n}p_{k}\ket{k}\bra{k}$, where $\{p_{k}\}_{k=1}^{n}$ is the probability distribution arranged in non increasing order. Due to the action of $V(t)$, the state vector
   $\vec{\rho}$ evolves around the vector $\vec{H}^{'}$ and when at $t=\tau$ it reaches the corresponding passive state vector $\vec{\rho}_{p}$ the external potential $V(t)$ gets switched off. A simple calculation exhibits that the extraction of work in this process is given by
  \begin{equation}\label{5}
  	W= Tr(\rho H)-Tr(\rho_{p} H)\\
  	 =\sum_{k,l=1}^{n}p_{k}\epsilon_{l}(|\langle\psi_{k}|l\rangle|^{2}-\delta_{kl}).
  \end{equation}
 The above equation quantifies the amount of extractable work under unitary evolution on a given quantum state mentioned as the first kind of work in the introduction and henceforth by "work" we will mean the same.
  See Fig \ref{fig2} to get a clear view for qubit scenario.
   \begin{figure}[htb]
  	\scalebox{0.75}{
  		\begin{tikzpicture}
  		\draw [->,black!100,thick](0,0)--(1.5,0.5);
  		\shade[ball color = red!50, opacity = 0.3] (0,0) circle (2.55cm);
  		\shade[ball color = yellow!50, opacity = 0.5] (0,0) circle (1.6cm);
  		\draw(0.75,1.040575)circle[x radius=0.923839cm, y radius=0.13cm, rotate=140];
  		\draw[->,dashed,thick](0,0)--(0.75*3,1.040575*3);
  		\node at (0.75*3+0.5,1.040575*3+0.5){\Large$\vec{H}+\vec{V}(t)$};
  		\node at (0.5,1.8){\Large$\rho_{p}$};
  		\draw[->,dashed,thick](0,1.5)--(0,0)--(0,-3);
  		\draw[->,thick](0,1.5)--(0,3);
  		  		\node at (1.7,0.7) {\Large$\vec{\rho}$};
  		\node at (0.4,2.9) {\Large$\vec{H}$};
  		\end{tikzpicture}}
  	\caption{(Color on-line) Time evolution of a state vector $\vec{\rho}$ along the circular path around the vector $\vec{H}+\vec{V}(t)$ is depicted. The rotation will occur on the surface of a smaller sphere of radius exactly equals to that of $\vec{\rho}$ to certify the evolution as a unitary. At time $\tau$, state $\rho$ reaches to $\rho_{p}$ and as a result energy difference amount of work can be extracted.}
  	\label{fig2}
  \end{figure}
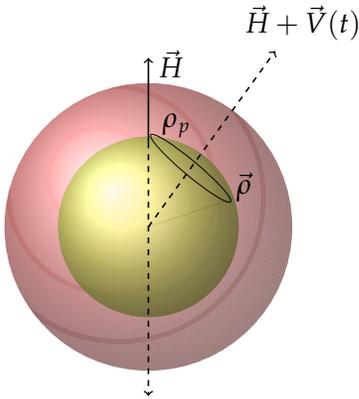
  \section{Cloning and Splitting of activity and energy contents}
 
  The amount of extractable work for any arbitrary quantum state characterizes the amount of energy stored in the state with respect to the state of minimum energy with identical spectrum. In the $\mathbb{R}^{3}$ Bloch ball representation all the states residing on the circumference of any circle (we call it the {\it circle of equal work}) of the plane perpendicular to the Hamiltonian axis and centering the axis itself, possess equal work content. 
  \begin{theorem}
  There is no universal work cloner for arbitrary quantum states.
  \end{theorem}
\begin{proof}
	Here we will show that the existence of a work cloning device for three different quantum states leads to a violation of the no-signaling principle. Although, our proof involves $\mathbb{C}^{2}$ Hilbert space it can be extended easily for arbitrary dimensional quantum states. To prove the above theorem, let's consider a device $\mathcal{D}$ which can clone the extractable work for atleast three pure qubits, namely $\ket{0}, \ket{1}$ and $\ket{+}=\frac{\ket{0}+\ket{1}}{\sqrt{2}}$. Now we will consider a bipartite scenario where a singlet state $\ket{\psi^{-}}$ is shared between two distant parties Alice and Bob. If Alice measures her qubit in any arbitrary basis then the reduced state of Bob's particle remains unaltered (i.e., $\frac{I}{2}$). However, if Bob's state changes with Alice's measurement choice, then he can measure his qubit and guess Alice's measurement with fidelity greater than $\frac{1}{2}$, which leads to the superluminal transmission of information.\par
	But the situation is different if Bob has access to the work cloning device $\mathcal{D}$. Suppose Alice measures her qubit in $\sigma_{z}$ basis. As a result Bob's reduced marginal will be $\rho_{1}=\frac{1}{2}\ketbra{0}{0}+\frac{1}{2}\ketbra{1}{1}$. Now the action of $\mathcal{D}$ on $\ket{0}$, or $\ket{1}$ will be exactly like cloning, as for these two states their circle of equal work is a unique point. So, after the action of $\mathcal{D}$ on Bob's side, the state becomes $\sigma_{1}=\frac{1}{2}\ketbra{00}{00}+\frac{1}{2}\ketbra{11}{11}$. On the other hand, if Alice makes a $\sigma_{x}$ measurement on her qubit, Bob's state will be $\rho_{2}=\frac{1}{2}\ketbra{+}{+}+\frac{1}{2}\ketbra{-}{-}$. Now, the action of $\mathcal{D}$ on $\ket{+}$ will be $\mathcal{D}:\ket{+}\to \frac{(\ket{0}+e^{-i\phi_{1}}\ket{1})}{\sqrt{2}}\otimes\frac{(\ket{0}+e^{-i\phi_{2}}\ket{1})}{\sqrt{2}}$, where $\phi_{1},\phi_{2}\in(0,2\pi]$. So, the action of $\mathcal{D}$ on $\rho_{2}$ will produce $\sigma_{2}= \frac{1}{8}(\ketbra{0}{0}+e^{-i\phi_{1}}\ketbra{1}{0}+e^{i\phi_{1}}\ketbra{0}{1}+\ketbra{1}{1})\otimes(\ketbra{0}{0}+e^{-i\phi_{2}}\ketbra{1}{0}+e^{i\phi_{2}}\ketbra{0}{1}+\ketbra{1}{1})+\frac{1}{2}\rho$, where, $\rho=\mathcal{D}\ketbra{-}{-}$ is an arbitrary bipartite quantum state. For our analysis, the action of $\mathcal{D}$ on $\ketbra{-}{-}$ is not required. However, independent of $\rho$  (but due to the fact that $\rho\geq0$), it is evident that due to the presence of other diagonal terms $\sigma_{2}\neq\sigma_{1}$. As a consequence it is possible for Bob to choose a proper measurement setting to discriminate optimally between these two states. This would imply the violation of no-signaling principle. Hence it is impossible to clone the amount of extractable work for arbitrary quantum states.
\end{proof}
However, it is possible to clone the energy content of a quantum state without maintaining its spectrum i.e., the entropy. 
 \begin{proposition}
	For any arbitrary quantum state with a given Hamiltonian, it is possible to clone its energy content.
\end{proposition}
\begin{proof}
	Given any arbitrary quantum state $\psi$ we can write it in a linear combination of energy eigen-basis. Explicitly, $\ket{\psi}=\sum_{k=1}^{d}c_{k}\ket{k}$, where $\{\ket{k}\}_{k=1}^{d}$ are the orthonormal energy eigen-basis. Hence there obviously exists a universal cloner $\mathbf{U}$ for this set of states, i.e., $\mathbf{U}\ket{k}\ket{0}=\ket{k}\ket{k}$. \par 
	The action of the global unitary can be written as $\mathbf{U}\ket{\psi}_{A}\ket{0}_{B}=\sum_{k=1}^{d}c_{k}\ket{k_{A}k_{B}}$. The reduced marginal of this state after the evolution will be $\rho_{A}=\rho_{B}=\sum_{k=1}^{d}|c_{k}|^{2}\ketbra{k}{k}$. Evidently, the energy corresponding to the marginals is given by $\sum_{k=1}^{d}|c_{k}|^{2}\epsilon_{k}$, where $\epsilon_{k}$'s are the energy eigen-values corresponding to the Hamiltonian $H_{A}$ and $H_{B}$. This is exactly equal to the energy of the initial state $\ket{\psi}$. The same unitary can clone the energy of any other quantum state also, in this sense $\mathbf{U}$ is a universal energy cloner.
\end{proof}
It is interesting to note that due to the energy-cloning possibilities of arbitrary pure quantum states, the same is true for their convex combinations also. Precisely, consider a quantum state $\rho=p\ketbra{\psi_{1}}{\psi_{1}}+(1-p)\ketbra{\psi_{2}}{\psi_{2}}$ and apply the same unitary $\mathbf{U}$ on the composite quantum state $\rho\otimes\ketbra{0}{0}$. The marginal states given by $\rho_{1}=\rho_{2}=\sum_{k=1}^{d}[p|c_{i}^{(1)}|^{2}+(1-p)|c_{i}^{(2)}|^{2}]\ketbra{i}{i}$ have energy $\sum_{k=1}^{d}[p|c_{i}^{(1)}|^{2}+(1-p)|c_{i}^{(2)}|^{2}]E_{i}$, which is exactly equal to the energy stored in the initial quantum state $\rho$.\par
In \cite{noinfosplit} it was shown that the information content of a single qubit can not be split in two different qubits. In other words, it is impossible to encode the information of a qubit contained in $\theta$ and $\phi$ in two different qubits simultaneously. As a consequence, from the thermodynamic perspective we can ask whether it is possible to split the energy content of a $d$-dimensional quantum state in two different qudits with arbitrary fraction? Obviously the action should be under an energy preserving unitary. In the following we will design such a unitary. 
\begin{proposition}
For any $d$-dimensional quantum state, energy-splitting is possible with any arbitrarily chosen fraction.
\end{proposition}
\begin{proof}
	Consider any arbitrary quantum state $\ket{\psi}=\sum_{k=0}^{d-1}c_{k}\ket{k}$ governed by the Hamiltonian $H=\sum_{k=1}^{d-1}\epsilon_{k}\ketbra{k}{k}$ (scaling the ground state energy to be zero). Therefore, the energy corresponding to the quantum state $\ket{\psi}$ is $\sum_{k=1}^{d-1}|c_{k}|^{2}\epsilon_{k}$.\\
	Now, in order to split the energy content of $\ket{\psi}$ into two different qudits in the ratio $\{p,(1-p)\}$, we choose the unitary action $\mathbf{U}\ket{k0}= \sqrt{p} \ket{k0}+\sqrt{1-p}\ket{0k}$. Then $\mathbf{U}\ket{\psi}_{S}\ket{0}_{A}= c_{0}\ket{0_{S}0_{A}}+\sum_{k=1}^{d-1}c_{k}(\sqrt{p}\ket{k_{S}0_{A}}+\sqrt{1-p}\ket{0_{S}k_{A}})$. Hence, the energy for  the system and ancillary qudit are $E_{S}=\Tr(\rho_{S}H)$ and $E_{A}=\Tr(\rho_{A}H)$ respectively, where $\rho_{S(A)}$ is the final reduced system (ancillary) state. It is easy to observe that $E_{S}=\sum_{k=1}^{d-1}p|c_{k}|^{2}\epsilon_{k}$ whereas, $E_{A}=\sum_{k=1}^{d-1}(1-p)|c_{k}|^{2}\epsilon_{k}$. Hence energy of the quantum state is split up in the required ratio.
\end{proof}
\section{Work Masking}
In quantum thermodynamics the difference between local and global extractable work termed as ergotropic gap plays a significant role in identifying the structure of a quantum state shared between its constituents. In particular, it is shown that separability of a bipartite quantum state invokes an upper bound on this quantity \cite{our'PRA}. The importance of those bipartite states for which local marginals are passive has also been studied. In this context, we ask the question whether it is possible to mask the work content of any arbitrary quantum state within the correlations of its bipartite extension. More precisely, given any arbitrary quantum state $\ket{\psi}$ along with a machine state, is it possible to design a unitary action such that the final bipartite state is locally passive in nature? We can also ask about the existence of a general quantum operation to perform the same work masking. 
\begin{lemma}
There is no universal work masking unitary which is energy preserving in nature.
\end{lemma}
\begin{proof}
We will demonstrate the proof for qubit case since the extension to higher dimensions is easy. Let's consider a pure qubit $\ket{\psi}=\alpha\ket{0}+\beta\ket{1}$ and be a machine state $\ket{0}$ of the same system. The energy preservation criterion on the global unitary demands $\ket{00}\to\ket{00}$, where within the $\ket{..}$ symbol, the first one denotes the system particle and the second one as the machine state. \par
Similarly, the action on $\ket{10}$ will be in general, $\ket{10}\to (a\ket{00}+b\ket{01}+c\ket{10}+d\ket{11})$. Now, the inner product preservation of the unitary demands $a=0$, which due to the energy preservation criterion implies that $d=0$. So action of the unitary will be $\mathbf{U}\ket{00}=\ket{00}$ and $\mathbf{U}\ket{10}=(b\ket{01}+c\ket{10})$. But for the final state produced under this operation, the marginals will have non-vanishing coherence in energy eigen basis, which will prevent them to be passive in nature.
\end{proof}
It remains open to identify the class of states for which the work masking is possible under energy conserving unitary evolution. For those states which are diagonal in energy eigen basis, it is possible to mask their energy in the correlation they will share in bipartite extension. In the following we will construct a unitary to do so. 
\begin{proposition}
	For any arbitrary mixed state diagonal with the governing equally spaced Hamiltonian, it is possible to mask its work content in the bipartite extension.
\end{proposition}
\begin{proof}
For any arbitrary $d$-dimensional quantum system governed by the linear Hamiltonian $H=\sum_{k=0}^{d-1}k\epsilon\ketbra{k}{k}$ we can construct a unitary acting globally on the system and machine state $\ket{0}$ of same dimension as follows  	 $\mathbf{U}:\ket{k0}\to\frac{\ket{0k}+\ket{1(k-1)}+...+\ket{k0}}{\sqrt{k+1}}, \forall_k\in\{0,...,(d-1)\}$. It is easy to check that the action of this unitary is energy preserving. Now let us consider the state  $\rho=\sum_{k=0}^{d-1}C_{k}\ketbra{k}{k}$, diagonal in energy basis, such that $\mathbf{U}(\sum_{k=0}^{d-1}C_{k}\ketbra{k}{k}\otimes\ketbra{0}{0})\mathbf{U}^{\dagger}=\sum_{k=0}^{d-1}C_{k}(\frac{\ket{0k}+\ket{1(k-1)}+...+\ket{k0}}{\sqrt{k+1}})(\frac{\bra{0k}+\bra{1(k-1)}+...+\bra{k0}}{\sqrt{k+1}})$. Due to the symmetric nature of the global system reduced marginals will be identical and given by $\sigma=\sum_{k=0}^{d-1}p_{k}\ketbra{k}{k}$, where $p_{k}=\sum_{j=k}^{d-1}\frac{C_{j}}{j+1}$. From the expression of $p_{k}$, evidently $p_{k+1}\le p_{k}$, whereas the energy corresponding to these two levels is such that $\epsilon_{k+1}\ge\epsilon_{k}$, thereby making the local marginals passive. It is obvious that due to the energy preserving constraint of the acting unitary, the work content of $\rho$ will be transferred to the final correlation between the system and and machine state.
\end{proof}
In the above we have shown using energy conserving unitary that it is not possible to mask the work content for an arbitrary quantum state. However for those states which are diagonal in the energy eigen basis it is possible to mask. Trivially we can extend this result for any arbitrary quantum state diagonal in a given basis, by transforming that basis in the energy eigen basis under unitary evolution. 
\begin{theorem}
	There is no universal work masking unitary.
\end{theorem}
\begin{proof}
We will prove the above statement for qubits which can be extended for higher dimensions. So, we want to prove the impossibility for the existence of a global unitary, $\mathbf{U}$, such that, $\forall\ket{\psi}\in\mathbb{C}^{2}, \mathbf{U}\ket{\psi}_{S}\ket{0}_{A}=\ket{\Psi}_{SA}$ and $\rho_{S/A}=\Tr_{A/S}(\ket{\Psi}\bra{\Psi})$ is passive in nature. The condition for local passivity restricts the local marginals to be diagonalized in energy eigen basis\cite{lenard}. Furthermore, for the bipartite state $\ket{\Psi}_{SA}$, the marginals will be of same eigen values, so the expected form for $\ket{\Psi}_{SA}$ is {\it either}, $c\ket{00}+d\ket{11}$, {\it or} $c\ket{01}+d\ket{10}$. However, the last one can be excluded, otherwise both the marginals can not be passive except for $|c|=|d|=\frac{1}{\sqrt{2}}$. So, the action on an arbitrary $\ket{\psi}=a\ket{0}+be^{i\phi}\ket{1}$ will be,\\
$\mathbf{U}(a\ket{0}+b e^{i\phi}\ket{1})\ket{0}=c\ket{00}+de^{i\tilde{\phi}}\ket{11}$,\\
 where in general $c$ and $\tilde{\phi}$ are functions of both $a$ and $\phi$. Local passivity of the final joint state demands that $\frac{1}{2}\leq|c|^{2}\leq1$. The action of same unitary on $\ket{\psi^{\perp}}\ket{0}$ will be $\mathbf{U}\ket{\psi^{\perp}}\ket{0}=\ket{\Psi^{\perp}}$, where $\ket{\Psi^{\perp}}=-d\ket{00}+e^{i\tilde{\phi}}c\ket{11}$. It is obvious that the local marginals for $\ket{\Psi^{\perp}}$ have exactly been flipped with respect to those of $\ket{\Psi}$. Therefore, both of them can not be locally passive simultaneously which implies that there is no universal work masking unitary.
\end{proof}
\textbf {Corollary 2.1. }{\it No qubit channel can map both the system and ancillary qubit along two Bloch radii.}
\begin{proof}
Let us suppose that there exists a joint unitary $\mathbf{V}$ acting on both the system and ancillary qubit as, $\mathbf{V}\ket{\psi^{k}_{S}}\ket{0_{A}}=\ket{\Psi^{k}_{SA}}$, such that, $\rho^{k}_{S}=p_{k}\ketbra{0}{0}+(1-p_{k})\ketbra{1}{1}$ and $\rho^{k}_{A}=q_{k}\ketbra{\phi}{\phi}+(1-q_{k})\ketbra{\bar{\phi}}{\bar{\phi}}$, where without loss of generality $\frac{1}{2}\le p_{k}, q_{k}\le 1$. But as the state $\ket{\Psi^{k}_{SA}}$ is pure, the spectrum for the marginals should be same and hence we assume $p_{k}=q_{k}$. Then we can apply another unitary $\mathbf{u}$ on the machine qubit such that, $\mathbf{u}:\{\phi, \bar{\phi}\}\to\{0,1\}$. As a consequence $(\mathbb{I}_{S}\otimes\mathbf{u}_{A})\mathbf{V}$ would be the combined unitary which takes the state $\ket{\psi^{k}_{S}}\ket{0_{A}}$ to $\sqrt{p_{k}}\ket{00}+\sqrt{1-p_{k}}\ket{11}$. This contradicts  Theorem 2 since the state so obtained is locally passive in nature. Note that the reduced dynamics on the system qubit for a global unitary can be realized as a qubit channel.
\end{proof} 
\textbf{Corollary 2.2. }{\it Quantum information can not be masked.} 
\begin{proof}
	Quantum no masking theorem \cite{nomasking}, which has a connection to the impossibilities of $(2,2)$ secret sharing in \cite{Cleve'PRL}, states that it is impossible to construct a unitary $\mathbf{U}$ such that $\mathbf{U}\ket{\psi_{S}^{(k)}}\ket{0_{A}}=\ket{\Psi_{SA}^{(k)}}$ where $\rho_{S(A)}=Tr_{A(S)}[\ketbra{\Psi_{SA}^{(k)}}{\Psi_{SA}^{(k)}}], \forall_k$, i.e., the quantum information of an arbitrary quantum state can not be masked in the correlations of extended bipartite system. However, Theorem 2 and Corollary 2.1 imply that it is impossible to mask the qubit along two Bloch radii in its bipartite extension. Hence quantum no masking turns out as an obvious consequence of our result.
\end{proof}
\textbf{Corollary 2.3.} {\it Existence of non-unitary work masker on the system-machine qubits.} 

\begin{proof}
The general action on the system and the machine state will be a CPTP map, which can be visualized as a reduced dynamics obtained from the action of a global unitary on the system, machine and some additional ancillas. To prove the generic existence of such a work masker, let us consider a global unitary $\mathbb{U}\ket{0}_{S}\ket{0}^{\otimes 3}_{A}=\ket{\phi^{+}}^{\otimes 2}$ and $\mathbb{U}\ket{1}_{S}\ket{0}^{\otimes 3}_{A}=\ket{\phi^{-}}^{\otimes 2}$, such that, $\mathbb{U}(a\ket{0}+b \ket{1})_{S}\ket{0}^{\otimes 3}_{A}= (a\ket{\phi^{+}}^{\otimes 2}+b\ket{\phi^{-}}^{\otimes 2})_{SA}$. Now a simple calculation can show that the marginals of joint state for each particles is $\frac{I}{2}$, irrespective of the values of $a$ and $b$. Hence the amount of work stored in the system is completely masked in the correlations shared among its constituents. This protocol can be easily related to error correcting codes \cite{errorshor,grassl}.
	\end{proof}
 \section{Conclusion}
 In summary, we have tried to characterize the impossibilities of certain thermodynamic operations in the quantum domain. More formally, several information theoretic {\it no-go} results in quantum theory are reformulated in terms of thermodynamic quantities to study their status from the perspective of quantum theory. Just like cloning of an unknown quantum state, copying its work content (which means a constraint equation on the parameters of the arbitrary quantum state) is also forbidden. This is one particular case of quantum partial cloning in the sense that between the qubit parameters $(\theta,\phi)$ the information in $\theta$ can not be cloned.
Further, it is interesting to characterize the parametric class of a general quantum state which can be cloned without disturbing any other fundamental principles.
 Although classical theory also exhibits no cloning, in the sense that given a single realization of a random variable it is impossible to figure out the probability distribution corresponding to that variable, classical broadcasting is possible. However, this opens up the possibility to clone the energy content of a given quantum state, which we have studied in Proposition 1. Furthermore, we have shown the impossibility of masking the work content of a quantum state, and it is observed that {\it quantum no masking} principle follows as an obvious consequence of it. Besides, the certainty to achieve the work masking with atleast four parties can be obtained from error correcting codes. It should be mentioned that our last theorem has some relevance from the perspective of qubit channels. More precisely, qubit channels can map the Bloch sphere to other different smooth convex regions inscribed inside the Bloch sphere itself \cite{braun}. However if we consider the dynamics of not only the system but also of the environment state, our result invokes a restriction on these allowed regions. Therefore, another direction of future study would be to modify the allowed domains of quantum channels while keeping a track of the evolved ancilla. Although our results deal with a particular kind of work stored in a quantum state, it opens up the possibilities to study other kinds of work as mentioned in the introduction. Furthermore, it is interesting to establish a connection of our results with some of the already existing physical principles, which may help to interpret these prohibitions from more fundamental perspectives.
 \section{Acknowledgment}
 M.A. would like to acknowledge the CSIR project 09/093(0170)/2016-EMR-I for financial support.
               
\end{document}